\documentclass[]{interact}

\usepackage{comment}

\theoremstyle{plain}
\newtheorem{theorem}{Theorem}[section]

\newtheorem{proposition}[theorem]{Proposition}

\newtheorem{assumption}[theorem]{Assumption}

\theoremstyle{definition}
\newtheorem{definition}[theorem]{Definition}

\theoremstyle{remark}
\newtheorem{remark}{Remark}

\newcommand{\ud}{\mathrm{d}}
\newcommand{\R}{\mathbb{R}}

\begin{document}


\title{Option Pricing in Illiquid Markets with Jumps}

\author{
\name{Jos\'{e} M. T. S. Cruz\textsuperscript{a}
and
Daniel \v{S}ev\v{c}ovi\v{c}\textsuperscript{b}\thanks{CONTACT J.~Cruz. Email: jmcruz@iseg.ulisboa.pt \newline to appear in: Applied Mathematical Finance }
}
\affil{\textsuperscript{a}
ISEG, University of Lisbon, Rua de Quelhas 6, 1200-781 Lisbon, Portugal; 
\textsuperscript{b}
Comenius University in Bratislava, Mlynsk\'a dolina, 84248 Bratislava, Slovakia}
}

\maketitle

\begin{abstract}
The classical linear Black--Scholes model for pricing derivative securities is a popular model in financial industry. It relies on several restrictive assumptions such as completeness, and frictionless of the market as well as the assumption on the underlying asset price dynamics following a geometric Brownian motion. The main purpose of this paper is to generalize the classical Black--Scholes model for pricing derivative securities by taking into account feedback effects due to an influence of a large trader on the underlying asset price dynamics exhibiting random jumps. The assumption that an investor can trade large amounts of assets without affecting the underlying asset  price itself is usually not satisfied, especially in illiquid markets. We generalize the Frey--Stremme nonlinear option pricing model for the case the underlying asset follows a L\'evy stochastic process with jumps. We derive and analyze a fully nonlinear parabolic partial-integro differential equation for the price of the option contract. We propose a semi-implicit numerical discretization scheme and perform various numerical experiments showing influence of a large trader and intensity of jumps on the option price. 
\end{abstract}

\begin{keywords}
Nonlinear partial integro-differential equation; L\'evy measure; Finite difference approximation  
\end{keywords}

\section{Introduction}
Over recent decades, the Black--Scholes model and its generalizations become widely used in financial markets because of its simplicity and existence of the analytic formula for pricing European style options. According to the classical theory developed by  Black,  Scholes and Merton the price  $V(t,S)$ of an option in a stylized financial market at time $t\in[0,T]$ and the underlying asset price $S$ can be computed as a solution to the linear Black--Scholes parabolic equation:
\begin{equation}
\label{eq6}
\frac{\partial V}{\partial t}(t,S) + \frac{1}{2}\sigma^2 S^2\frac{\partial^2 V}{\partial S^2}(t,S) + r S\frac{\partial V}{\partial S}(t,S) - r V(t,S) = 0,\mbox{ }t \in [0,T),S > 0.
\end{equation}
Here $\sigma>0$ is the historical volatility of the underlying asset driven by the geometric Brownian motion, $r > 0$ is the risk-free interest rate of zero-coupon bond. A solution is subject to the terminal pay-off condition $V(T,S) = \Phi(S)$ at maturity $t = T$. Evidence from stock markets observations indicates that this model is not the most realistic one, since it assumes that the market is liquid, complete, frictionless and without transaction costs.  We also recall that the linear Black--Scholes equation provides a solution corresponding to a perfectly replicated  portfolio which need not be a desirable property. In the last two decades some of these assumptions have been relaxed in order to model, for instance, the presence of transaction costs (see e.g.  Kwok \cite{NBS5} and Avellaneda and Paras \cite{NBS7}), feedback and illiquid market effects due to large traders choosing given stock-trading strategies (Sch\"onbucher and Willmott \cite{NBS13}, Frey and Patie \cite{NBS11}, Frey and Stremme \cite{NBS10}), risk from the unprotected portfolio (Janda\v{c}ka and \v{S}ev\v{c}ovi\v{c} \cite{NBS1}).  In all aforementioned generalizations of the linear Black--Scholes equation (\ref{eq6}) the constant volatility $\sigma$ is replaced by a nonlinear function $\tilde\sigma(S \partial _S^2 V)$ depending on the second derivative $\partial _S^2 V$ of the option price itself. In the class of generalized Black--Scholes equation with such a  nonlinear diffusion function, an important role is played by the nonlinear Black--Scholes model derived by Frey and Stremme in \cite{NBS1} (see also \cite{NBS11},\cite{Frey98}). In this model the asset dynamics takes into account the presence of feedback effects due to a large trader choosing his/her stock-trading strategy (see also \cite{NBS13}). The diffusion coefficient is again non-constant:
\begin{equation}
\tilde\sigma(S \partial^2_S V)^2 = \sigma^2 \left(1-\varrho S\partial^2_S V\right)^{-2},
\label{doplnky-c-frey}
\end{equation}
where $\sigma, \varrho>0$ are constants.

Another important direction in generalizing the original Black--Scholes equation arises from the fact that the sample paths of a Brownian motion are continuous, but the realized stock price of a typical company exhibits random jumps over the intraday scale, making the price trajectories discontinuous. In the classical Black--Scholes model the underlying asset price process is assumed to follow a  geometric Brownian motion. However the empirical distribution of stock returns exhibits fat tails. Several alternatives have been proposed in the literature for the generalization of this model. The models with jumps and diffusion can, at least in part, solve the problems inherent to the linear Black--Scholes model and they have also an important role in the options market. While in the Black--Scholes model the market is complete, implying that every pay-off can be perfectly replicated, in jump--diffusion models there is no perfect hedge and this way the options are not redundant. It turns out that the option price can be computed from the solution $V(t,S)$ of the following partial integro-differential (PIDE) Black--Scholes equation:
\begin{eqnarray}
&&\frac{\partial V}{\partial t}(t,S)+\frac{1}{2}\sigma^2 S^2 \frac{\partial^2 V}{\partial S^2 }(t,S) +r S\frac{\partial V}{\partial S}(t,S)-r V(t,S)\nonumber
\\
&&+\int_{\mathbb{R}} V(t,S+H(z,S))-V(t,S)-H(z,S)\frac{\partial V}{\partial S}(t,S)\nu(\ud z)=0,\label{nonlinearPIDE_one-intro}
\end{eqnarray}
where $H(z,S) = S (e^z-1)$ and $\nu$ is the so-called L\'evy measure characterizing the underlying asset process with random jumps in time and space. Note that, if  $\nu=0$ then (\ref{nonlinearPIDE_one-intro}) reduces to the classical linear Black--Scholes equation (\ref{eq6}). 

The novelty and main purpose of this paper is to take into account both directions of generalizations of the Black--Scholes equation. The assumption that an investor can trade large amounts of the underlying asset without affecting its price is no longer true, especially in illiquid markets. Therefore, we will derive, analyze, and perform numerical computation of the model. We relax the assumption of liquid market following the Frey--Stremme model under the assumption that the underlying asset price follows a L\'evy stochastic process with jumps. We will show that the corresponding PIDE nonlinear equation has the form: 
\begin{eqnarray}
&&\frac{\partial V}{\partial t}+\frac{1}{2}\frac{\sigma^2}{\left(1-\varrho S\partial_S \phi\right)^{2}} S^2 \frac{\partial^2 V}{\partial S^2 } +r S\frac{\partial V}{\partial S}-rV\nonumber
\\
&&+\int_{\mathbb{R}} V(t,S+H(t,z,S))-V(t,S)-H(t,z,S)\frac{\partial V}{\partial S}\nu(\ud z)=0,\label{nonlinearPIDE_one-intro-nonlin}
\end{eqnarray}
where the function $H(t,z,S)$ may depend e.g. on the large trader strategy function $\phi=\phi(t,S)$. This function  may depend on the delta $\partial_S V$ of the price $V$, if $\varrho>0$. 

The paper is organized as follows. In the next Section 2 we recall known facts regarding exponential L\'{e}vy models. We also recall important classes of L\'evy measures with finite and infinite activity. Section 3 is devoted to derivation of the novel option pricing model taking into account feedback effects of a large trader on the underlying asset following jump-diffusion L\'evy process. We show that the price of an option can be computed from a solution to a fully nonlinear partial integro-differential equation (PIDE) (\ref{nonlinearPIDE_one-intro-nonlin}). We also derive a formula for the trading strategy function $\phi$ which minimizes the variance of the tracking error.  Next, in Section 4 we present a semi-implicit numerical discretization scheme for solving the resulting nonlinear PIDE. The scheme is based on finite difference approximation. In Section 5 we present numerical results when considering the Variance--Gamma process. We also perform sensitivity numerical analysis of a solution with respect to model parameters $\rho$ and $\nu$. 

\section{Preliminaries, definitions and motivation}

We consider a stylized economy with two traded assets, a risky asset, usually a stock with a price $S_{t}$, and a risk-less asset, typically a bond with a price $B_{t}$ which is taken as numeraire. The bond market is assumed to be perfectly elastic as bonds are assumed to be more liquid when compared to stocks. In this economy there are two type of traders: the reference traders and the program traders. The program traders are also referred to as portfolio insurers since they use dynamic hedging strategies to hedge portfolio against jumps in stock prices. They are either single traders or a group of traders acting together. It is assumed that their trades influence the stock price equilibrium. The reference traders can be considered as representative traders of many small agents. We assume they act as price takers. Typically, it is assumed that $\tilde{D}(t,Y_{t},S_{t})$ is the reference trader demand function which depends on the income process $Y_t$ or some other fundamental state variable influencing the reference trader demand. The aggregate demand of program traders is denoted by $\varphi(t,S_{t})=\xi \phi(t,S_{t})$, where $\xi$ is the number of  written identical securities that the program traders are trying to hedge and $\phi(t,S_{t})$ is the demand per unit of the security being hedged. For simplicity, we assume that $\xi$ is the same for every program trader. The general case where different securities  are considered can be found e.g. in \cite{SirPapa98}. Assume the supply of a stock with the price  $\tilde{S_{0}}$ is constant. Let $D(t,Y,S)=\frac{\tilde{D}(t,Y,S)}{\tilde{S_{0}}}$ denote  the quantity demanded by a reference trader per unit of supply. Then the total demand relative to the supply at time $t$ is given by $G(t,Y,S)=D(t,Y,S)+\rho\phi(t,S)$, where $\rho=\frac{\xi}{\tilde{S_{0}}}$ and $\rho\phi(t,S)$ is the proportion of the total supply of the stock that is being traded by program traders. In order to obtain the market equilibrium the variables $Y$ and $S$ should satisfy $G(t,Y,S)=1$. Assume that the function $G$ is monotone  with respect to $Y$ and $S$ variables, and it is sufficiently smooth. Then we can solve the implicit equation $G(t,Y_t,S_t)=1$ to obtain $S_{t}=\psi(t,Y_{t})$ where $\psi$ is a sufficiently smooth function. Following  \cite{SirPapa98}, we assume that the stochastic process $Y_{t}$ has the following dynamics: 
\[
\ud Y_{t}=\mu(t,Y_{t})\ud t+\eta(t,Y_{t}) \ud W_{t}.
\]
Then, by using It\^{o}'s lemma for the process $S_{t}=\psi(t,Y_{t})$ we obtain
\begin{equation}
\ud S_t = \left( \partial_t \psi + \mu\partial_y\psi +\frac{\eta^2}{2}\partial^2_y\psi \right) \ud t + \eta \partial_y\psi  \ud W_t \equiv b(t,S_t) S_t \ud t + v(t,S_t) S_t \ud W_t.
\label{BS-v}
\end{equation}
It means that $S_t$ follows a geometric Brownian motion with a nonconstant volatility function $v(t,S)=\eta(t,Y)\partial_Y\psi(t,Y)/\psi(t,Y)$ where $Y=\psi^{-1}(t,S)$. 
Following the argument used in the derivation of the original Black--Scholes equation we obtain a generalization of the Black--Scholes partial differential equation with a nonconstant volatility function $\sigma=v(t,S)$.
In this paper we follow Frey and Stremme's approach  (cf. \cite{NBS11, Frey98}). The idea is to prescribe a dynamics for the underlying stock price instead of deriving it by using the market equlibrium and dynamics for the income process $Y_t$ as it is done e.g. in \cite{SirPapa98}. This way Frey and Stremme derived the same stock price dynamics as in \cite{SirPapa98} corresponding to a situation where the demand function is of logarithmic type, $D(Y,S)=\ln(\frac{Y^{\gamma}}{S})$, where $\gamma=\frac{\sigma}{\eta_0}$, and the income process $Y_t$ follows a geometric Brownian motion, i.e.
\begin{eqnarray}
&&\partial_Y D(Y,S)=\gamma\frac{1}{Y},\ \partial_S D (Y,S)=-\frac{1}{S},\ \  \ud Y_{t}=\mu_0 Y_{t}\ud t+\eta_0 Y_{t} \ud W_{t},
\label{generalPDE}
\\
&&v(t,S)=\eta(t,Y)\frac{\partial_Y\psi(t,Y)}{\psi(t,Y)} = - \frac{\eta_0 Y}{S}\frac{\gamma\frac{1}{Y}}{-\frac{1}{S}+\rho \frac{\partial \phi}{\partial S}}=\frac{\sigma}{1-\rho S\frac{\partial \phi}{\partial S}}.
\nonumber
\end{eqnarray}
Assuming the delta hedging strategy with $\phi(t,S)=\partial_S V(t,S)$ and inserting the volatility function $v(t,S)$ into (\ref{BS-v}) we obtain the generalized Black--Scholes equation with the nonlinear diffusion function of the form (\ref{doplnky-c-frey}).

Our main goal is to extend the Frey--Stremme  model to an underlying asset following a L\'{e}vy process. Next, we recall basic properties of L\'evy jump-diffusion processes.

\subsection{Exponential L\'{e}vy models}\label{sectiontwo}

Let $X_t, t\ge0,$ be a stochastic process. The measure $\nu(A)$ of a Borel set $A \in \mathcal{B}(\mathbb{R})$ defined  by $\nu(A)=\mathbb{E}\left[ J_{X}([0,1]\times A)\right]$ where $J_{X}([0,t]\times A)=\# \left\{s \in [0,t]:\Delta X_{s} \in A \right\}$ is the Poisson random measure. It gives the mean number, per unit of time, of jumps  whose amplitude belongs to the set $A$. Recall that the L\'{e}vy-It\^{o} decomposition provides a representation of $X_t$ which can be interpreted as a combination of a Brownian motion with a drift $\omega$ and an  infinite sum of independent compensated Poisson processes with variable jump sizes (see \cite{ConTan03}), i.e. 
\[
\ud X_{t}=\omega \ud t+\sigma \ud W_{t}+ \int_{\left|x\right|\geq 1}x J_{X}\left(\ud t,\ud x\right)+ \int_{\left|x\right|<1}x \widetilde{J}_{X}\left(\ud t,\ud x\right), 
\]
where $\widetilde{J}_{X}\left(\left[0,t\right]\times A\right) =J_{X}\left(\left[0,t\right]\times A\right)-t \nu\left(A\right)$ is the compensation of $J_{X}$.

Any L\'{e}vy process is a strong Markov process, the associated semigroup is a convolution semigroup. Its infinitesimal generator $L:u\mapsto L[u]$ is a nonlocal  partial integro-differential operator given by (see \cite{App04}):
\begin{eqnarray}
L[u](x)&=&\lim_{h \to 0^+} \frac{\mathbb{E}\left[u\left(x+X_{h}\right)\right]-u\left(x\right)}{h}
\nonumber \\
&=&\frac{\sigma^{2}}{2}\frac{\partial^{2}u}{\partial x^{2}}+\gamma\frac{\partial u}{\partial x}+\int_{\mathbb{R}}\left[u\left(x+y\right)-u\left(x\right)-y 1_{\left|y\right| \leq 1}\frac{\partial u}{\partial x}\left(x\right)\right] \nu(\ud y), 
\label{eq:infgen no condition}
\end{eqnarray}
which is well defined for any compactly supported function $u\in C^{2}_0\left(\mathbb{R}\right)$.

Let $S_{t}, t\ge 0,$ be a stochastic process representing an underlying asset process under a filtered probability space $\left(\Omega,\mathcal{F},\left\{\mathcal{F}_{t}\right\},\mathbb{P}\right)$. The filtration $\left\{\mathcal{F}_{t}\right\}$ represents the price history up to the time $t$. If the market is arbitrage-free then there is an equivalent  measure $\mathbb{Q}$ under which discounted prices of all traded financial assets are $\mathbb{Q}-$ martingales. This result is known as the fundamental theorem of asset pricing $($see \cite{ConTan03}$)$. The measure $\mathbb{Q}$ is also known as the risk neutral measure. We consider the exponential L\'{e}vy model in which the risk-neutral price process $S_{t}$ under $\mathbb{Q}$ is given by $S_{t}=e^{rt+X_{t}}$, where $X_{t}$ is a L\'{e}vy process under $\mathbb{Q}$ with the characteristic triplet $\left(\sigma,\gamma,\nu\right)$. Then the arbitrage-free market hypothesis imposes that $\widehat{S}_{t}=S_{t}e^{-rt}=e^{X_{t}}$ is a martingale, which is equivalent to the following conditions imposed on the triplet $\left(\sigma,\gamma,\nu\right)$:

\begin{equation}
\int_{\left|y\right|\ge 1} e^{y} \nu (\ud y)<\infty,\ \ \gamma\in\mathbb{R},\ \  \gamma=-\frac{\sigma^{2}}{2}-\int_{-\infty}^{+\infty} \left(e^{y}-1-y1_{\left|y\right|\leq 1}\right) \nu(\ud y). \label{expcond}
\end{equation}
The risk-neutral dynamics of $S_{t}$ under $\mathbb{Q}$ is given by
\begin{equation}
\ud S_{t}=r S_t \ud t+ \sigma S_t \ud W_{t}+ \int_{\mathbb{R}}\left(e^{y}-1\right)S_{t}\widetilde{J}_{X}\left(dt,dy\right).
\end{equation}
The exponential price process $e^{X_{t}}, t\ge 0,$ is also a Markov process with the state space $\left(0,\infty\right)$ and the infinitesimal generator:
\begin{eqnarray}
L^{S}[V](S)&=&\lim_{h \rightarrow 0} \frac{\mathbb{E}[V(S e^{X_{h}})]-V(S)}{h} = r S\frac{\partial V}{\partial S}+\frac{\sigma^{2}}{2}S^2\frac{\partial^{2}V}{\partial S^{2}}\\
&& +\int_{\mathbb{R}}\left[V(S e^{y})-V(S)-S (e^{y}-1)\frac{\partial V}{\partial S}\right]\nu(\ud y)\label{eq: infgenS }
\end{eqnarray}
(see \cite{ConTan03}). Recall that a  L\'{e}vy process is called the L\'{e}vy type stochastic integral if it has the following representation:
\[
\ud X_{t}=\omega \ud t+\sigma \ud W_{t}+\int_{|x|\ge 1} K(t,x) J_{X}(\ud t, \ud x)
+\int_{|x|<1} H(t,x) \tilde{J}_{X}(\ud t, \ud x).
\]
An important result that will be needed later is the following variant of It\^{o}'s lemma.

\begin{theorem}\cite{App04}
Let $f\in C^{1,2}([0,T]\times\mathbb{R})$ and  $H, K\in C([0,T]\times\mathbb{R})$. Let $X_{t}, t\ge 0,$ be a L\'{e}vy stochastic process. Then 
\begin{eqnarray}
\ud f(t,X_{t})&=&\frac{\partial f}{\partial t}\ud t+\frac{\partial f}{\partial x}\ud X_{t}+\frac{1}{2}\frac{\partial^2 f}{\partial x^2}\ud [X_{t},X_{t}]\nonumber
\\
&&+\int_{|x|\ge 1} f(t,X_{t}+K(t,x))-f(t,X_{t})J_{X}(\ud t, \ud x)\label{itolemma}
\\
&&+\int_{|x|<1} f(t,X_{t}+H(t,x))-f(t,X_{t})\tilde{J}_{X}(\ud t, \ud x)\nonumber
\\
&&+\int_{|x|<1} f(t,X_{t}+H(t,x))-f(t,X_{t})-H(t,x)\frac{\partial f}{\partial x}(t,X_t)\nu(\ud x)\ud t .
\nonumber
\end{eqnarray}
\end{theorem}
A classical example of a L\'{e}vy process is the jump-diffusion model introduced by Merton in \cite{Merton76}. It has the following dynamics:
\begin{eqnarray}
&& \ud X_{t}=\left(b+\int_{|x|<1} x \nu(\ud x)\right)\ud t+\int_{|x|\ge 1} x J_{X}(\ud t, \ud x) +\sigma \ud W_{t}+\int_{|x|<1} x \tilde{J}_{X}(\ud t, \ud x).
\nonumber
\end{eqnarray}
Then, by applying It\^{o}'s lemma to $S_{t}=e^{X_{t}}$ we obtain:
\begin{eqnarray}
&&\ud S_{t}=(b+\sigma^{2}/2)S_{t} \ud t+\sigma S_{t}\ud W_{t}+S_{t}\int_{\mathbb{R}} \left(e^{x}-1\right) J_{X}(\ud t, \ud x) .
\nonumber
\end{eqnarray}
In financial applications, exponential L\'{e}vy models are of several types. In this paper we are concerned with the so-called jump-diffusion models in which we represent the log-price as a L\'{e}vy process with a non-zero diffusion part $(\sigma >0)$ and a jump process with either finite activity with ($\nu(\mathbb{R})<\infty$) or infinite activity ($\nu(\mathbb{R})=\infty$). 

In the context of financial modelling a jump-diffusion model was proposed by Merton in \cite{Merton76}. The random jump variables are normally distributed with the mean $m$ and variance $\delta^2$. Its L\'{e}vy density is given by:
\begin{equation}
\nu(\ud x)=\lambda \frac{1}{\delta\sqrt{2\pi}}e^{-\frac{(x-m)^2}{2\delta^2}}\ud x\,.
\label{merton-density}
\end{equation}
Another popular and frequently used model is the so-called double exponential model which was introduced by Kou in \cite{Kou2002}. In this model the distribution of jumps  have a L\'evy measure of the form:
\begin{equation}
\nu (\ud x)=\lambda \left( \theta \lambda^{+} e^{-\lambda^{+} x}1_{x>0}+ (1-\theta) \lambda^{-} e^{\lambda^{-} x}1_{x<0}\right)\ud x,
\label{double-density}
\end{equation}
where $\lambda$ is the intensity of jumps, $\theta$ is the probability of having a positive jump and  $\lambda^\pm> 0$ correspond to the level of decay of the distribution of  positive and negative jumps. This implies that the distribution of jumps is asymmetric and the tails of the distribution of returns are semi-heavy. 

Among examples of infinite activity L\'evy processes used in the financial modelling there are e.g. the Variance Gamma (see \cite{DPE98}), Normal Inverse Gaussian (NIG) (see \cite{BARNIE01}), or CGMY processes. The Variance Gamma process is a process of infinite activity $\nu(\mathbb{R})=\infty$ and finite variation $\int_{|x|\leq 1}|x|\nu(\ud x)<\infty$. Its L\'{e}vy measure is given by:
\begin{equation}
\nu (\ud x)=\frac{1}{\kappa \left|x\right|}e^{Ax-B\left|x\right|}\ud x \ \text{ with }\  A=\theta/\sigma^{2} \ \text{ and }\  B=\sigma^{-2}\sqrt{\theta^2+2\sigma^2/\kappa},
\label{vargamma-density}
\end{equation} 
where the parameters $\sigma$ and $\theta$ are related to volatility and drift of the Brownian motion and $\kappa$ is a parameter related to the variance of the subordinator, in this case the Gamma process (see \cite{ConTan03}). 

All aforementioned examples of jump-diffusion models have L\'evy measures belonging to the class of the so-called admissible activity L\'evy measures. 

\begin{definition}\label{def-admissiblemeasure}
A L\'evy measure $\nu$ is called an admissible activity L\'evy measure if  
\begin{equation}
0\le \frac{\nu(\ud z)}{\ud z} 
\le 
h(z)\equiv C |z|^{-\alpha}\left(e^{D^{-}z}1_{z\geq 0}+e^{D^{+}z}1_{z< 0}\right)e^{-\mu z^{2}},
\label{growth_measure}
\end{equation}
for any $z\in\R$ and shape parameters $\alpha\geq 0$, $D^{\pm}\in\mathbb{R}$ and $\mu\geq 0.$
\end{definition}

\begin{remark}
Note that the additional conditions $\int_{\mathbb{R}} \min(z^2,1) \nu (\ud z)<\infty$  and 
$\int_{\left|z\right|>1} e^{z} \nu (\ud z)<\infty$ are satisfied provided that $\nu$ is an admissible L\'evy measure with shape parameters $\alpha<3$, and, either $\mu>0, D^\pm\in \mathbb{R}$, or $\mu=0$ and $D^-+1<0<D^+$.

For the Merton model we have $\alpha=0, D^\pm=0$ and $\mu=1/(2\delta^2)>0$. In the Kou model $\alpha=\mu=0, D^+=\lambda^-, D^-=-\lambda^+$. As for the Variance Gamma process we have $\alpha=1, \mu=0, D^\pm=A\pm B$.
\end{remark}

\section{Feedback effects under jump-diffusion underlying asset price dynamics}

Let us suppose that a large trader uses a stock-holding strategy $\alpha_{t}$ and $S_{t}$ is a cadlag process (right continuous  with limits to the left). Henceforth,  we shall identify $S_{t}$ with $S_{t^-}$. We assume $S_t$ has the following  dynamics:
\begin{eqnarray}
&&\ud S_{t}=\mu S_{t}\ud t+\sigma S_{t} \ud W_{t}+\rho S_{t} \ud \alpha_{t}+ \int_{\mathbb{R}} S_{t}(e^{x}-1) J_{X}(\ud t, \ud x).
\label{Sdynamicsimplicit}
\end{eqnarray}
It can be viewed as a perturbation of the classical jump-diffusion model. Indeed, if a large trader does not trade then $\alpha_{t}=0$ or the market liquidity parameter $\rho$ is set to zero then the stock price $S_t$ follows the classical jump-diffusion model.

In what follows, we will assume the following structural hypothesis:
\begin{assumption}\label{assumptionone}
Assume the trading strategy $\alpha_{t}=\phi(t,S_{t})$ and the parameter $\rho\ge 0$ satisfy $\rho L<1$, where $L=\sup_{S>0} |S\frac{\partial \phi}{\partial S}|$.
\end{assumption}

Next we show an explicit formula for the dynamics of $S_{t}$ satisfying (\ref{Sdynamicsimplicit}) under certain regularity assumptions made on the stock-holding function $\phi(t,S)$.
\begin{proposition}
Suppose that the stock-holding strategy $\alpha_{t}=\phi(t,S_{t})$  satisfies Assumption \ref{assumptionone} where $\phi\in C^{1,2}([0,T]\times \mathbb{R^{+}})$. If the process $S_{t}, t\ge 0,$ satisfies the implicit stochastic equation  (\ref{Sdynamicsimplicit}) then the process $S_{t}$ satisfies the following SDE:
\begin{eqnarray}
\ud S_{t}=b(t,S_{t})S_{t}\ud t+v(t,S_{t})S_{t} \ud W_{t}+\int_{\mathbb{R}} H(t,x,S_{t}) J_{X}(\ud t, \ud x)
\label{Sdynamicsexplicit},
\end{eqnarray}
where
\begin{eqnarray}
&&b(t,S)=\frac{1}{1-\rho S\frac{\partial \phi}{\partial S}(t,S)}\left(\mu+\rho\left(\frac{\partial \phi}{\partial t}+\frac{1}{2}v(t,S)^2S^2 \frac{\partial^2 \phi}{\partial S^2}\right)\right),
\label{coeffcondb}
\\
&&v(t,S)=\frac{\sigma}{1-\rho S\frac{\partial \phi}{\partial S}(t,S)},
\label{coeffcondv}
\\
&&H(t,x,S)=S_{}(e^{x}-1)+\rho S\left[\phi(t,S+H(t,x,S))-\phi(t,S)\right].
\label{coeffcondone}
\end{eqnarray}

\end{proposition}
\begin{proof}
We can rewrite the SDE  \eqref{Sdynamicsexplicit} for  $S_{t}$, in the following way:
\begin{eqnarray}
\ud S_{t}&=&\left(b(t,S_{t}) S_{t} +\int_{|x|<1} H(t,x,S_{t})\nu(\ud x)\right)\ud t + v(t,S_{t}) S_{t} \ud W_{t}\nonumber
\\
&&+\int_{|x|\ge 1} H(t,x,S_{t}) J_{X}(\ud t, \ud x)+\int_{|x|<1} H(t,x,S_{t}) \tilde{J}_{X}(\ud t, \ud x).
\nonumber
\end{eqnarray}
Since $\phi(t,S)$ is assumed to be a smooth function then, by applying It\^{o} formula \eqref{itolemma} to the process $\phi(t,S_{t})$, we obtain
\begin{eqnarray}
\ud \alpha_{t}&=&\left(\frac{\partial \phi}{\partial t}+\frac{1}{2}v(t,S_{t})^2S_{t}^2 \frac{\partial^2 \phi}{\partial S^2}\right)\ud t+\frac{\partial \phi}{\partial S}\ud S_{t}
\\
&&+\int_{\mathbb{R}} \phi(t,S_{t}+H(t,x,S_{t}))-\phi(t,S_{t})-H(t,x,S_{t}) \frac{\partial \phi}{\partial S}(t,S_t) J_{X}(\ud t, \ud x).
\nonumber
\end{eqnarray}
Now, inserting the differential  $\ud \alpha_{t}$ into \eqref{Sdynamicsimplicit}, we obtain
\begin{eqnarray}
\ud S_{t}&=&\mu S_{t}\ud t+\sigma S_{t} \ud W_{t}+\int_{\mathbb{R}} S_{t}(e^{x}-1) J_{X}(\ud t, \ud x)+\rho S_{t}\frac{\partial \phi}{\partial S}\ud S_{t}\nonumber
\\
&&+\rho S_{t} \left(\frac{\partial \phi}{\partial t}+\frac{1}{2}v(t,S_{t})^2S_{t}^2 \frac{\partial^2 \phi}{\partial S^2}\right)\ud t\label{Sdynamicsexplicitone}
\\
&&+\rho S_{t}\int_{\mathbb{R}} \phi(t,S_{t}+H(t,x,S_{t}))-\phi(t,S_{t})- H(t,x,S_{t}) \frac{\partial \phi}{\partial S}(t,S_t) J_{X}(\ud t, \ud x).
\nonumber
\end{eqnarray}
Rearranging terms in (\ref{Sdynamicsexplicitone}) we conclude 
\begin{eqnarray}
&&(1- \rho S_{t}\frac{\partial \phi}{\partial S}(t,S_{t}))\ud S_{t}=(\mu S_{t}+\rho S_{t}(\frac{\partial \phi}{\partial t}+\frac{1}{2}v(t,S_{t})^2 S_{t}^2 \frac{\partial^2 \phi}{\partial S^2}))\ud t 
\nonumber
\\
&&+\sigma S_{t} \ud W_{t}-\rho S_{t}\int_{\mathbb{R}} H(t,x,S_{t}) \frac{\partial \phi}{\partial S}(t, S_t)  J_{X}(\ud t, \ud x)
\label{Sdynamicsexplicittwo}
\\
&&+\int_{\mathbb{R}} S_{t}(e^{x}-1)+\rho S_{t}\left( \phi(t,S_{t}+H(t,x,S_{t}))-\phi(t,S_{t})\right) J_{X}(\ud t, \ud x).
\nonumber
\end{eqnarray}
Comparing terms in (\ref{Sdynamicsexplicit}) and  (\ref{Sdynamicsexplicittwo}) we end up with expressions \eqref{coeffcondb}, \eqref{coeffcondv}, and the implicit equation for the function $H$:
\begin{eqnarray}
H(t,x,S)&=&\frac{1}{1-\rho S\frac{\partial \phi}{\partial S}(t,S)}\left(S(e^{x}-1)+\rho S\left(\phi(t,S+H(t,x,S))-\phi(t,S)\right)\right)\nonumber
\\
&&-\frac{1}{1-\rho S\frac{\partial \phi}{\partial S}(t,S)}\rho S\frac{\partial \phi}{\partial S}(t,S) H(t,x,S).
\label{coeffcondH}
\end{eqnarray}
Simplifying this expression for $H$ we conclude  \eqref{coeffcondone}, as claimed.
\end{proof}

The function $H$ is given implicitly by equation \eqref{coeffcondone}. If we expand its solution $H$ in terms of a small parameter $\rho$, i.e. $H(t,x,S) = H^0(t,x,S) + \rho H^1(t,x,S) + O (\rho^2)$ as $\rho\to0$, we conclude the following proposition:

\begin{proposition}\label{prop-Happrox}
Assume $\rho$ is small. Then the first order approximation of the function $H(t,x,S)$ reads as follows:
\begin{eqnarray}
&&H(t,x,S)= S(e^{x}-1)+\rho S\left(\phi(t,Se^{x})-\phi(t,S)\right)+O (\rho^2)\quad\text{as}\ \rho\to0.
\label{Hsmall}
\end{eqnarray}
\end{proposition}

\begin{proposition}
Assume that the asset price process $S_{t}=e^{X_t +r t}$ fulfills SDE \eqref{Sdynamicsexplicit} where the L\'evy measure $\nu$ is such that $\int_{|x|\ge 1}  e^{2x} \nu\left(\ud x\right)<\infty$. Denote by $V(t,S)$ the price of a  derivative security given by 
\begin{equation}
V(t,S)=\mathbb{E}\left[e^{-r(T-t)}\Phi(S_{T})|S_{t}=S\right]=e^{-r(T-t)}\mathbb{E}\left[\Phi(Se^{r(T-t)+X_{T-t}})\right].
\end{equation}
Assume that the pay-off function $\Phi$ is a Lipschitz continuous function and the function $\phi$ has a bounded derivative. Then $V(t,S)$ is a solution to the PIDE:
\begin{eqnarray}
&&\frac{\partial V}{\partial t}+\frac{1}{2}v(t,S)^2 S^2\frac{\partial^2 V}{\partial S^2}+rS\frac{\partial V}{\partial S}-rV\nonumber
\\
&&\qquad +\int_{\mathbb{R}} V(t,S+H(t,x,S))-V(t,S)-H(t,x,S)\frac{\partial V}{\partial S}(t,S)\nu(\ud x)=0,
\label{nonlinearPIDE}
\end{eqnarray}
where $v(t,S)$ and $H(t,x,S)$ are given by $\eqref{coeffcondv}$ and $\eqref{coeffcondone}$, respectively.
\end{proposition}

\begin{proof}
 The asset price dynamics of $S_{t}$ under the $\mathbb{Q}$ measure is given by
\begin{eqnarray}
&&\ud S_{t}=rS_{t}\ud t + v(t,S_{t}) S_{t} \ud W_{t}+\int_{\mathbb{R}} H(t,x,S_{t}) \tilde{J}_{X}(\ud t, \ud x).\label{SQdynamics}
\end{eqnarray}
If we apply It\^{o}'s lemma to $V(t,S_{t})$ we obtain $\ud(V(t,S_{t})e^{-rt})=a(t)\ud t+\ud M_{t}$ where 
\begin{eqnarray*}
a(t)&=&\frac{\partial V}{\partial t}+\frac{1}{2}v(t,S_t)^2S_{t}^2 \frac{\partial^2 V}{\partial S^2}+rS_{t}\frac{\partial V}{\partial S}-rV\nonumber
\\
&&+\int_{\mathbb{R}} V(t,S_{t}+H(t,x,S_{t}))-V(t,S_{t})-H(t,x,S_{t})\frac{\partial V}{\partial S}(t,S_t)\nu(\ud x),
\\
\ud M_{t}&=&e^{-rt}S_{t}v(t,S_{t})\frac{\partial V}{\partial S}\ud W_{t}+e^{-rt}\int_{\mathbb{R}} V(t,S_{t}+H(t,x,S_{t}))-V(t,S_{t})\tilde{J}_{X}(\ud t, \ud x).
\end{eqnarray*}
Our goal is to show that $M_{t}$ is a martingale. Consequently, we have $a\equiv 0$ a.s., and  $V$ is a solution to \eqref{nonlinearPIDE} (see Proposition 8.9 of \cite{ConTan03}). To prove the term  $\int_{0}^{T}e^{-rt}\int_{\mathbb{R}} V(t,S_{t}+H(t,x,S_{t}))-V(t,S_{t})\tilde{J}_{S}(\ud t, \ud y)$ is a martingale it is sufficient to show that 
\begin{eqnarray}
&&\mathbb{E}\left[\int_{0}^{T}e^{-2rt} \left(\int_{\mathbb{R}}V(t,S_{t}+H(t,x,S_{t}))-V(t,S_{t})\nu(\ud x)\right)^2\ud t\right]<\infty.
\end{eqnarray}
Since $\sup_{0\le t\le T}\mathbb{E}\left[e^{X_{T-t}}\right] <\infty$ and the pay-off function $\Phi$ is Lipschitz continuous, $V(t,S)$ is Lipschitz continuous as well with some Lipschitz constant $C>0$. As the function $\phi(t,S)$ has bounded derivatives  we obtain
\[
S \left|\phi(t,S+H(t,x,S))-\phi(t,S)\right|
\leq S \left|\frac{\partial \phi}{\partial S}\right| |H(t,x,S)|
\leq L |H(t,x,S)|
\]
(see Assumption \ref{assumptionone}).  Since $H(t,x,S)=S (e^{x}-1)+\rho S (\phi(t,S+H(t,x,S))-\phi(t,S) )$ we obtain $|H(t,x,S)|^{2}\leq S^{2}(e^{x}-1)^2/(1-\rho L)^2$.  As $V$ is Lipschitz continuous with the Lipschitz constant $C>0$ we have
\begin{eqnarray*}
&&\mathbb{E}\left[\int_{0}^{T}e^{-2rt} \left(\int_{\mathbb{R}}V(t,S_{t}+H(t,x,S_{t}))-V(t,S_{t})\nu(\ud x)\right)^2\ud t\right]
\\
&&\le \frac{C^2 }{(1-\rho L)^2} \mathbb{E}\left[\int_{0}^{T}\int_{\mathbb{R}} e^{-2rt}  |S_t|^2  (e^{x}-1)^2 \nu(\ud x)  \ud t \right] <\infty,
\end{eqnarray*}
because $\sup_{t\in \left[0,T\right]} \mathbb{E}\left[S_{t}^2\right]<\infty$. Here $C_0 = \int_{\mathbb{R}} (e^{x}-1)^2 \nu(\ud x) <\infty$ due to the assumptions made on the measure $\nu$. It remains to prove that $\int_{0}^{T}e^{-rt}S_{t}v(t,S_{t})\frac{\partial V}{\partial S}(t,S_t)\ud W_{t}$ is a martingale. Since $S \frac{\partial \phi}{\partial S}(t,S)$ is assumed to be bounded we obtain
\begin{eqnarray}
0<v(t,S)=\frac{\sigma}{1-\rho S\frac{\partial \phi}{\partial S}(t,S)}\leq \frac{\sigma}{1-\rho L}\equiv C_{1}<\infty .\nonumber
\end{eqnarray}
Therefore $\mathbb{E}[\int_{0}^{T}e^{-2rt}(\frac{\partial V}{\partial S}(t,S_{t})v(t,S_{t}) S_{t})^{2}\ud t]
\leq C^{2} C_1^2 \int_{0}^{T}e^{-2rt}\mathbb{E}[S^{2}_{t}]\ud t<\infty$ because $S_t$ is a martingale. Hence $M_t$ is a martingale as well. As a consequence, $a\equiv 0$ and so $V$ is a solution to PIDE (\ref{nonlinearPIDE}), as claimed.
\end{proof}

\begin{remark}
If $\rho=0$ then $H(t,x,S)=S (e^{x}-1)$ and equation (\ref{nonlinearPIDE}) reduces to:
\begin{equation}
\frac{\partial V}{\partial t}+\frac{\sigma^2}{2}S_{}^2\frac{\partial^2 V}{\partial S^2}+rS_{}\frac{\partial V}{\partial S}-rV+\int_{\mathbb{R}} V(t,S e^{x})-V(t,S_{})-S (e^{x}-1)\frac{\partial V}{\partial S}(t,S) \nu(\ud x)=0,
\label{classicalPIDE}
\end{equation}
which is the well-known classical PIDE. If there are no jumps ($\nu=0$) and a trader follows the delta hedging strategy, i.e.  $\phi(t,S)=\partial_S V(t,S)$, then equation (\ref{nonlinearPIDE}) reduces to the Frey--Stremme option pricing model:
\begin{equation}
\frac{\partial V}{\partial t}+\frac{1}{2}\frac{\sigma^2}{\left(1-\varrho S\partial^2_S V\right)^{2}} S^2 \frac{\partial^2 V}{\partial S^2 } +r S\frac{\partial V}{\partial S}-rV =0
\label{Frey}
\end{equation}
(cf. \cite{Frey98}). Finally, if $\rho=0$ and $\nu=0$ equation (\ref{nonlinearPIDE}) reduces to the classical linear Black--Scholes equation. 
\end{remark}

For simplicity, we assume the interest rate is zero, $r=0$. Then the function $V(t,S)$ is a solution to the PIDE:
\begin{eqnarray}
\frac{\partial V}{\partial t}&+&\frac{1}{2}v(t,S)^2 S^2 \frac{\partial^2 V}{\partial S^2}
\nonumber 
\\
&+&\int_{\mathbb{R}} V(t,S+H(t,x,S))-V(t,S)-H(t,x,S)\frac{\partial V}{\partial S}(t,S)\nu(\ud x)=0.
\label{equationsimple}
\end{eqnarray}
Let us define the tracking error of a trading strategy $\alpha_t=\phi(t,S_t)$ as follows:
$e_{T}^{M}:=\Phi(S_T)-V_0=V(T, S_{T})-V_{0}-\int_{0}^{T} \alpha_{t} \ud S_{t}$.

By applying It\^{o}'s formula to $V(t,S_{t})$ and using $\eqref{equationsimple}$ we obtain
\begin{eqnarray}
&&V(T,S_{T})-V_0= V(T,S_{T})-V(0,S_{0})=\int_{0}^{T} \ud V(t,S_{t})
\nonumber
\\
&&=\int_{0}^{T} \frac{\partial V}{\partial S}\ud S_{t}+\int_{0}^{T} \frac{\partial V}{\partial t}+\frac{1}{2}v(t,S_{t})^2 S_{t}^2\frac{\partial^2 V}{\partial^2 S} \ud t\nonumber
\\
&&\qquad +\int_{0}^{T}\int_{\mathbb{R}} V(t,S_{t}+H(t,x,S_{t}))-V(t,S_{t})-H(t,x,S_{t})\frac{\partial V}{\partial S} J_{X}(\ud t, \ud x)\nonumber
\\
&&=\int_{0}^{T} \frac{\partial V}{\partial S}\ud S_{t} - \int_{0}^{T} \int_{\mathbb{R}} V(t,S_{t}+H(t,x,S_{t}))-V(t,S_{t})-H(t,x,S_{t})\frac{\partial V}{\partial S}\nu(\ud x)\ud t\nonumber
\\
&&\qquad +\int_{0}^{T}\int_{\mathbb{R}} V(t,S_{t}+H(t,x,S_{t}))-V(t,S_{t})-H(t,x,S_{t})\frac{\partial V}{\partial S} J_{X}(\ud t, \ud x)\nonumber
\\
&&=\int_{0}^{T} \frac{\partial V}{\partial S}\ud S_{t}+\int_{0}^{T}\int_{\mathbb{R}} V(t,S_{t}+H(t,x,S_{t}))-V(t,S_{t})-H(t,x,S_{t})\frac{\partial V}{\partial S} \tilde{J}_{X}(\ud t, \ud x).
\nonumber
\end{eqnarray}

Using expression (\ref{SQdynamics}) for the dynamics of the asset price $S_t$ (with $r=0$), the tracking error $e_{T}^{M}$ can be expressed as follows:
\begin{eqnarray}
e_{T}^{M} &=& V(T, S_{T})-V_{0}-\int_{0}^{T} \alpha_{t} \ud S_{t} = \int_{0}^{T} \left(\frac{\partial V}{\partial S}(t,S_t) - \alpha_t\right)\ud S_{t}
\nonumber\\
&& +\int_{0}^{T}\int_{\mathbb{R}} V(t,S_{t}+H(t,x,S_{t}))-V(t,S_{t})-H(t,x,S_{t})\frac{\partial V}{\partial S} \tilde{J}_{X}(\ud t, \ud x)
\nonumber
\\
&=& \int_{0}^{T}  v(t, S_t) S_t \left(\frac{\partial V}{\partial S}-\alpha_t\right)\ud W_{t}
\label{trackingerr}
\\
&& +\int_{0}^{T}\int_{\mathbb{R}} V(t,S_{t}+H(t,x,S_{t}))-V(t,S_{t})-\alpha_t H(t,x,S_{t}) \tilde{J}_{X}(\ud t, \ud x).
\nonumber
\end{eqnarray}

\begin{remark}
For the delta hedging strategy $\alpha_t=\phi(t,S_{t})=\frac{\partial V}{\partial S}(t,S_{t})$ the tracking error function $e_{T}^{M}$ can be expressed as follows:
\[
e_{T}^{M}=\int_{0}^{T}\int_{\mathbb{R}} V(t,S_{t}+H(t,x,S_{t}))-V(t,S_{t})-H(t,x,S_{t})\frac{\partial V}{\partial S}(t, S_t) \tilde{J}_{X}(\ud t, \ud x).
\]
Clearly, the tracking error for the delta hedging strategy need not be zero for $\nu\not\equiv0$. 
\end{remark}

Next, we propose a criterion that can be used to find the optimal hedging strategy.

\begin{proposition}\label{prop-trackingerror-min}
The trading strategy $\alpha_t=\phi(t,S_t)$ of a large trader minimizing the variance $\mathbb{E}\left[(\epsilon_{T}^M)^{2}\right]$ of the tracking error is given by the implicit equation:
\begin{eqnarray}
\phi(t,S_t)&=&\beta^\rho(t,S_t)\bigl[v(t,S_{t})^{2}S_{t}^{2}\frac{\partial V}{\partial S}(t,S_{t})
\nonumber
\\
&& +\int_{\mathbb{R}} \left(V(t,S_{t}+H(t,x,S_{t}))-V(t,S_{t})) H(t,x,S_{t}\right)\nu(\ud x)\bigr],
\label{strategy}
\end{eqnarray}
where $\beta^\rho(t,S_t)=1/[v(t,S_{t})^{2}S_{t}^{2}+\int_{\mathbb{R}}H(t,x,S_{t})^2\nu(\ud x) ]$ and $H(t,x,S)= S (e^x-1) + \rho S [\phi(t,S + H(t,x,S) )-\phi(t,S)]$.
\end{proposition}

\begin{proof}
Using expression (\ref{trackingerr}) for the tracking error $\epsilon_{T}^M$ and It\^{o}'s isometry we obtain 
\begin{eqnarray*}
\mathbb{E}\left[(\epsilon_{T}^M)^{2}\right]
&=&\mathbb{E}\left[\int_{0}^{T}v(t,S_{t})^{2}S_{t}^{2}\left(\frac{\partial V}{\partial S}(t,S_t) - \alpha_{t}\right)^{2}\ud t\right]\nonumber
\\
&&+\mathbb{E}\left[\int_{0}^{T}\int_{\mathbb{R}}\left(V(t,S_{t}+H(t,x,S_{t}))-V(t,S_{t})-\alpha_{t}H(t,x,S_{t})\right)^{2}\nu(\ud x)\ud t\right].
\end{eqnarray*}
The minimizer $\alpha_{t}$ of the above convex quadratic minimization problem satisfies the first order necessary conditions 
$\ud (\mathbb{E}\left[\epsilon_{T}^{2}\right],\alpha_{t})=0$, that is, 
\begin{eqnarray}
0&=&-2\mathbb{E}\left[\int_{0}^{T}\left(v(t,S_{t})^{2}S_{t}^{2}\left(\frac{\partial V}{\partial S}(t,S_t) -\alpha_{t}\right)\right.\right.\nonumber
\\
&&\left.\left.+\int_{\mathbb{R}}H(t,x,S_{t})\bigl(V(t,S_{t}+H(t,x,S_{t}))-V(t,S_{t})-\alpha_{t}H(t,x,S_{t})\bigr)\nu(\ud x)\right)\omega_{t}\ud t\right]\nonumber
\end{eqnarray}
for any variation $\omega_{t}$. Thus the tracking error minimizing strategy $\alpha_{t}$ is given by \eqref{strategy}.
\end{proof}

\begin{remark}
The optimal trading strategy minimizing the variance of the tracking error need not satisfy the structural Assumption~\ref{assumptionone}. For instance, if $\nu=0$ then the tracking error minimizer is just the delta hedging strategy $\phi=\partial_S V$. In the case of a call or put option its gamma, i.e. $\partial^2_S V(t,S)$ becomes infinite as $t\to T$ and $S=K$. Given a level $L>0$ we can however minimize the tracking error $\mathbb{E}\left[\epsilon_{T}^{2}\right]$ under the additional constraint $\sup_{S>0} |S\frac{\partial \phi}{\partial S}(t,S)|\le L$. That is we can solve the following convex constrained nonlinear optimization problem
\[
\min_{\phi} \ \ \mathbb{E}\left[\epsilon_{T}^{2}\right] \ \ s.t. \ \ |S\partial_S \phi|\le L
\]
instead of the unconstrained minimization problem proposed in Proposition~\ref{prop-trackingerror-min}.
\end{remark}

\begin{remark}
Notice that, if $\nu=0$ and $\rho\ge 0$, the trading strategy $\alpha_t$ reduces to the Black--Scholes delta hedging strategy, i.e. $\alpha_{t}=\frac{\partial V}{\partial S}(t,S_t)$. If $\nu\not\equiv0$ and $\rho=0$, then the optimal trading strategy becomes $\alpha_t=\phi^{0}(t,S_t)$ where
\[
\phi^{0}(t,S_t)=\beta^0(t,S_t) \left(\sigma^{2}S_{t}^{2}\frac{\partial V}{\partial S}(t,S_t) + \int_{\mathbb{R}}S_{t}(e^{x}-1)\left(V(t,S_{t}e^{x})-V(t,S_{t})\right)\nu(\ud x)\right),
\]
where $\beta^0(t,S_t)=1/[\sigma^{2}S_{t}^{2}+\int_{\mathbb{R}}S^2_{t}(e^{x}-1)^2\nu(\ud x)]$. 
\end{remark}

We conclude this section by the following proposition providing the first order approximation of the tracking error minimizing trading strategy for the case when the parameter $\rho\ll 1$ is small. In what follows, we derive the first order approximation of $\phi^\rho(t,S_{t})$ in the form  $\phi^\rho(t,S_{t})=\phi^{0}(t,S_t) + \rho\phi^1(t,S_t)+O(\rho^{2})$ as $\rho\to 0$.

Clearly,  the first order Taylor expansion for the volatility function $v(t,S)$ has the form:
\[
v(t,S)^2 = \frac{\sigma^2}{(1-\rho S\partial_S \phi)^2} = \sigma^2 + 2 \rho  \sigma^2 S \frac{\partial\phi^0}{\partial S}(t,S) + O(\rho^{2}), \ \ \text{as}\ \rho\to 0.
\]
With regard to Proposition~\ref{prop-Happrox} (see \eqref{Hsmall}) we have  $H(t,x,S) = H^0(t,x,S) + \rho H^1(t,x,S) + O(\rho^{2})$, where
\begin{equation}
H^0(t,x,S) = S (e^x-1), \qquad H^1(t,x,S) = S [\phi^0(t,S e^x) - \phi^0(t,S)].
\label{Hexpand}
\end{equation}
The function $\beta^\rho$ can be expanded as follows: $\beta^\rho(t,S) = \beta^0(t,S) + \rho\beta^{(1)}(t,S) + O(\rho^{2})$,
\begin{eqnarray}
\beta^0(t,S) &=& 1/[\sigma^{2}S^{2}+ S^2 \int_{\mathbb{R}}(e^{x}-1)^2\nu(\ud x)], \quad 
\label{betaexpand}
\\
\beta^{(1)}(t,S) &=& - (\beta^0(t,S))^2 
\left[2 \sigma^{2}S^{3} \frac{\partial\phi^0}{\partial S}(t,S) + 2 S^2\int_{\mathbb{R}}(e^{x}-1) [\phi^0(t,S e^x) - \phi^0(t,S)]\nu(\ud x)\right].
\nonumber
\end{eqnarray}
Using the first order expansions of the functions $v^2, \beta^\rho$ and $H$ we obtain the following results.

\begin{proposition}
For small values of the parameter $\rho\ll 1$, the tracking error variance minimizing strategy $\alpha_{t}=\phi^\rho(t,S_{t})$ is given by 
\begin{eqnarray}
&&\phi^\rho(t,S_{t})=\phi^{0}(t,S_t) + \rho\phi^{(1)}(t,S_t)+O(\rho^{2}), \ \ \text{as}\ \rho\to 0, 
\end{eqnarray}
where 
\begin{eqnarray*}
\phi^{(1)}(t,S) &=&\beta^0(t,S) \left[ 2 \sigma^2 S^3 \frac{\partial V}{\partial S}(t,S) \frac{\partial \phi^0}{\partial S}(t,S)\right.
\\
&& + \left. \int_{\mathbb{R}}\left(V(t,S e^{x})-V(t,S) +  \frac{\partial V}{\partial S}(t,S e^x) H^0(t,x,S) \right) H^1(t,x,S) \nu(\ud x)
\right] 
\\
&& + \beta^{(1)}(t,S) \left[ \sigma^2 S^2 \frac{\partial V}{\partial S}(t,S)  
 +  \int_{\mathbb{R}}\left(V(t,S e^{x})-V(t,S)\right) H^0(t,x,S) \nu(\ud x)
\right]
\end{eqnarray*}
and the functions $H^0, H^1, \beta^0$ and $\beta^{(1)}$ are defined as in \eqref{Hexpand} and \eqref{betaexpand}.
\end{proposition}

\section{Implicit-explicit numerical discretization scheme}

The aim of this section is to propose a full time-space discretization scheme for solving the nonlinear PIDE (\ref{nonlinearPIDE}). The method of discretization is based on a finite difference approximation of all derivatives occurring in (\ref{nonlinearPIDE}) and approximation of the integral term by means of the trapezoidal  integration rule  on a truncated domain. 

In order to solve \eqref{nonlinearPIDE} we transform it into a nonlinear parabolic PIDE defined on $\mathbb{R}$. Indeed, using the standard transformations $V(t,S)=e^{-r\tau}u(\tau,x),\phi(t,S)=\psi(\tau,x)$ where $\tau=T-t,x=\ln(\frac{S}{K})$ we conclude that $V(t,S)$ is a solution to \eqref{nonlinearPIDE} if and only if the function $u(\tau,x)$ solves the following nonlinear parabolic equation:

\begin{eqnarray}
\frac{\partial u}{\partial \tau}&=&\frac{1}{2}\frac{\sigma^2}{(1-\rho \frac{\partial\psi}{\partial x})^2}
\frac{\partial^2 u}{\partial^2 x}
+\left(r-\frac{1}{2}\frac{\sigma^2}{(1-\rho \frac{\partial\psi}{\partial x})^2}\right)\frac{\partial u}{\partial x}
\label{nonlinearPIDEaltsimplified}
\\
&+&\int_{\mathbb{R}} u(\tau,x+\xi(\tau,z,x))-u(\tau,x)-H(T-\tau,z,K e^{x})\frac{1}{K} e^{-x}\frac{\partial u}{\partial x}(\tau,x) \nu(\ud z), 
\nonumber
\\
&&u(0,x)=h(x) \equiv \Phi(K e^{x}), \quad (\tau,x)\in [0,T]\times \mathbb{R},
\end{eqnarray}
and
\begin{eqnarray}
&&H(t,z,S)=S(e^{z}-1)+\rho S [\phi(t, S+H(t,z,S))-\phi(t,S) ],
\label{Hfunction}
\\
&&\xi(\tau, z,x)=\ln(1+ K^{-1} e^{-x} H(T-\tau, z, K e^x)).
\label{xibig}
\label{boundaryPIDE}
\end{eqnarray}

\subsection{A numerical scheme for solving  nonlinear PIDEs with finite activity L\'evy measures}
We first consider the case when the  L\'evy measure $\nu$ has a finite activity, i.e. $\nu(\mathbb{R})<\infty$. Let us denote 
\[
\lambda=\int_{\mathbb{R}}\nu(\ud z),
\quad\text{and}\quad 
\omega(\tau, x)=\int_{\mathbb{R}}H(T-\tau,z,S_{0}e^{x})\frac{1}{S_{0}}e^{-x} \, \nu(\ud z).
\]
We have $\lambda<\infty$. Observe that   \eqref{nonlinearPIDEaltsimplified}  is equivalent to 
\begin{equation}
\frac{\partial u}{\partial \tau}=\frac{1}{2}\frac{\sigma^2}{(1-\rho \frac{\partial\psi}{\partial x})^2}
\frac{\partial^2 u}{\partial^2 x}
+\left(r-\frac{1}{2}\frac{\sigma^2}{(1-\rho \frac{\partial\psi}{\partial x})^2} -\omega \right)\frac{\partial u}{\partial x}
-\lambda u+\int_{\mathbb{R}} u(\tau,x+\xi(\tau,z,x))\nu(\ud z).
\label{nonlinearPIDEaltsimplifiedAlgo}
\end{equation}
We proceed to solve \eqref{nonlinearPIDEaltsimplifiedAlgo}  by means of the semi-implicit finite difference scheme proposed in \cite{Vol05}. The idea is to separate the right-hand side into two parts:  the differential part and the integral part.

Let $u_{i}^{j}= u(\tau_{j},x_{i}), \tau_{j}=j \Delta \tau, x_{i}=z_i=i\Delta x$ for $i=-N+1,\cdots,N-1$ and $j=1,\cdots ,M$. We approximate the differential part implicitly except of $\psi(\tau,x)$
\begin{eqnarray*}
&&\left(\frac{\partial u}{\partial x}\right)^{j}_{i}\approx \left\{
\begin{array}{rl}
\frac{u^{j+1}_{i+1}-u^{j+1}_{i}}{\Delta x}, & \text{if } \frac{(\sigma^j_i)^2}{2}-r+ \omega^j_i  <0,
\\
\frac{u^{j+1}_{i}-u^{j+1}_{i-1}}{\Delta x}, & \text{if } \frac{(\sigma^j_i)^2}{2}-r+ \omega^j_i \geq 0,
\end{array} \right.
\ \ \sigma^j_i = \frac{\sigma}{1-\rho D\psi_{i}^{j}}, 
\\
&&\left(\frac{\partial^{2}u}{\partial x^{2}}\right)^{j}_{i}\approx \frac{u^{j+1}_{i+1}-2u^{j+1}_{i}+u^{j+1}_{i-1}}{(\Delta x)^2},
\\
&&\left(\frac{\partial u}{\partial \tau}\right)^{j}_{i}\approx\frac{u_{i}^{j+1}-u_{i}^{j}}{\Delta t}.
\quad 
\left(\frac{\partial \psi}{\partial x}\right)^{j}_{i}\approx\frac{\psi_{i+1}^{j}-\psi_{i}^{j}}{\Delta x}=D\psi_{i}^{j}.
\end{eqnarray*}
As for the integral operator, first we have to truncate the integration domain to a bounded interval $[B_{l},B_{r}]$. We approximate this integral by choosing integers $K_{l}$ and $K_{r}$ such that $[B_{l},B_{r}]\subset [(K_{l}-1/2)\Delta x,(K_{l}+1/2)\Delta x].$  
Then
\begin{eqnarray}
&&\int_{B_{l}}^{B_{r}} u(\tau_{j},x_{i}+\xi(\tau_j, z_{i},x_{i})) \,\nu(\ud z)\approx \sum_{k=K_{l}}^{K_{r}} u(\tau_{j},x_{i}+\xi(\tau_j, z_{k},x_{i}))\nu_{k},
\end{eqnarray}
where $\nu_{k}=\frac{1}{2} \left(\nu(z_{k+1/2})+\nu(z_{k-1/2}) \right) \Delta x$. Analogously, 
\[
\omega^j_i \approx \frac{e^{-x_i}}{K} \sum_{k=K_{l}}^{K_{r}} H(T-\tau_{j},z_{k},K e^{x_{i}})\nu_{k}, \quad \text{and}\ \  \lambda \approx \sum_{k=K_{l}}^{K_{r}} \nu_{k},
\]
where $\xi(\tau, z,x)$ is given as in \eqref{xibig}.

Inserting the finite difference approximations of derivatives of $u$ into $\eqref{nonlinearPIDEaltsimplifiedAlgo}$ we obtain
\begin{eqnarray}
\frac{u_{i}^{j+1}-u_{i}^{j}}{\Delta t}&=&\frac{1}{2}(\sigma^j_i)^2
\frac{u^{j+1}_{i+1}-2u^{j+1}_{i}+u^{j+1}_{i-1}}{(\Delta x)^2}
-\lambda u_{i}^{j+1}
\label{PIDERBFIllSimplified_discone}
\\
&&+(r-\frac{1}{2}(\sigma^j_i)^2-\omega^j_i )\frac{u^{j+1}_{i+1}-u^{j+1}_{i}}{\Delta x} + \sum_{k=K_{l}}^{K_{r}} u(\tau_{j},x_{i}+\xi(\tau_j,  z_{k},x_{i}))\nu_{k},
\nonumber
\end{eqnarray}
provided that $\frac{1}{2} (\sigma^j_i)^2-r+\omega^j_i<0$. Similarly, we can derive a difference equation for the case $\frac{1}{2} (\sigma^j_i)^2-r+\omega^j_i\ge 0$. If we define coefficients $\beta^j_{i\pm}$, and $\beta^j_{i}$  as follows: 
\begin{eqnarray}
\beta^j_{i\pm}&=&-\frac{\Delta \tau}{2(\Delta x)^2}(\sigma^j_i)^2
-\frac{\Delta \tau}{\Delta x}\left(r-\frac{1}{2}(\sigma^j_i)^2-\omega^j_i\right)^\pm,
\label{betaplusminus}
\\
\beta^j_{i}&=&1+\Delta \tau \lambda - (\beta^j_{i-} + \beta^j_{i+}),
\label{beta0}
\end{eqnarray}
where $(a)^+=\max(a,0)$, $(a)^-=\min(a,0)$. Then the tridiagonal system of linear equations for the solution $u^j=(u^j_{-N+1},\cdots, u^j_{N-1})^T, \ j=0, \cdots,  M$, reads as follows:

\begin{eqnarray}
u_{i}^{0}&=&h(x_{i}), \ \text{for}\ i=-N+1,\cdots, N-1,\nonumber
\\
u_{i}^{j+1}&=&g(\tau_{j+1},x_{i}), \ \text{for}\  i=-N+1, \cdots, -N/2 -1,
\nonumber
\\
\beta^j_{i+}u_{i+1}^{j+1}+\beta^j_{i}u_{i}^{j+1}+\beta^j_{i-}u_{i-1}^{j+1}&=&u_{i}^{j}+\Delta\tau  \sum_{k=K_{l}}^{K_{r}} u(\tau_{j},x_{i}+\xi(\tau_j, z_{k},x_{i}))\nu_{k} 
\label{PIDERBFIllSimplified_discthree}, 
\\
&& \quad  \text{for}\  i=-N/2+1,\cdots,N/2-1,
\nonumber
\\
u_{i}^{j+1}&=&g(\tau_{j+1},x_{i}), \ \text{for}\ i=N/2, \cdots, N -1,
\nonumber
\end{eqnarray}
where 
\[
\xi(\tau_j, z_{k},x_{i})=\ln(1 + K^{-1} e^{-x_i} H(T-\tau_{j},z_{k}, K e^{x_{i}})),
\]
and $g$ is a function of points $x_i$ lying outside the localization interval. Following Proposition 4.3.1 in \cite{Vol05} the recommended choice is $g(\tau,x)=h(x+r\tau)=\Phi(K e^{r\tau + x} )$. The term $u(\tau_{j},x_{i}+\xi(\tau_j, z_{k},x_{i}))$ entering the sum in the right-hand side of \eqref{PIDERBFIllSimplified_discthree} is approximated by means of the first order Taylor series expansion:
\[
u(\tau_{j},x_{i}+\xi(\tau_j, z_{k},x_{i}))\approx u^j_i + \frac{u_{i+1}^{j}-u_{i}^{j}}{\Delta x} \xi(\tau_j, z_{k},x_{i}).
\]

\subsection{Numerical scheme for solving  nonlinear PIDEs with infinite activity L\'evy measures}

Next we consider the case when the L\'evy measure has an infinite activity, e.g. the Variance Gamma process where its L\'{e}vy density explodes at zero and $\nu(\mathbb{R})=\infty$. Equation  \eqref{nonlinearPIDEaltsimplified} is equivalent to 
\begin{eqnarray}
\frac{\partial u}{\partial \tau}&=&\frac{1}{2}\frac{\sigma^2}{(1-\rho \frac{\partial\psi}{\partial x})^2}
\frac{\partial^2 u}{\partial^2 x}
+\left(r-\frac{1}{2}\frac{\sigma^2}{(1-\rho \frac{\partial\psi}{\partial x})^2} -\omega \right)\frac{\partial u}{\partial x}
\nonumber
\\
&& +\int_{\mathbb{R}} u(\tau,x+\xi(\tau,z,x))-u(\tau,x)\nu(\ud z).
\label{nonlinearPIDEaltsimplifiedAlgoInfinite}
\end{eqnarray}
Equation \eqref{nonlinearPIDEaltsimplifiedAlgoInfinite} differs from \eqref{nonlinearPIDEaltsimplifiedAlgo} as the term $u(\tau,x)$ is contained in the integral part because $\lambda=\int_{\mathbb{R}} \nu(\ud z) = \infty$. Proceeding similarly as for discretization of \eqref{nonlinearPIDEaltsimplifiedAlgo} we can solve  \eqref{nonlinearPIDEaltsimplifiedAlgoInfinite} numerically by means of a semi-implicit finite difference scheme. If the coefficients $\beta^j_{i\pm}$ are defined as in \eqref{betaplusminus} and $\beta^j_{i}=1  - (\beta^j_{i-} + \beta^j_{i+})$, then the solution vector $u^j=(u^j_{-N+1},\cdots, u^j_{N-1})^T, \ j=0, \cdots, M$, is a solution to the following tridiagonal system of linear equations:

\begin{eqnarray}
u_{i}^{0}&=&h(x_{i}), \ \text{for}\ i=-N+1,\cdots, N-1,\nonumber
\\
u_{i}^{j+1}&=&g(\tau_{j+1},x_{i}), \ \text{for}\  i=-N+1, \cdots, -N/2 -1,
\label{PIDERBFIllSimplified_discthreeInfinity}
\\
\beta^j_{i+}u_{i+1}^{j+1}+\beta^j_{i}u_{i}^{j+1}+\beta^j_{i-}u_{i-1}^{j+1}&=&u_{i}^{j}+\Delta\tau  \sum_{k=K_{l}}^{K_{r}} \left( u(\tau_{j},x_{i}+\xi(\tau_j, z_{k},x_{i})) - u(\tau_j, x_i) \right)\nu_{k}, 
\nonumber
\\
&& \quad  \text{for}\  i=-N/2+1,\cdots,N/2-1,
\nonumber
\\
u_{i}^{j+1}&=&g(\tau_{j+1},x_{i}), \ \text{for}\ i=N/2, \cdots, N -1.
\nonumber
\end{eqnarray}
The term $u(\tau_{j},x_{i}+\xi(\tau_j, z_{k},x_{i})) - u(\tau_j, x_i)$ entering the sum in the right-hand side of \eqref{PIDERBFIllSimplified_discthreeInfinity} is again approximated by means of the first order Taylor series expansion, i.e. 
\[
u(\tau_{j},x_{i}+\xi(\tau_j, z_{k},x_{i})) - u(\tau_j, x_i) \approx \frac{u_{i+1}^{j}-u_{i}^{j}}{\Delta x} \xi(\tau_j, z_{k},x_{i}).
\]

\section{Numerical results}

In this section we present results of numerical experiments using the finite difference scheme described in Section 4 for the case of a European put option, i.e. $\Phi(S)=(K-S)^+$. As for the L\'evy process we consider the Variance Gamma process with parameters $\theta=-0.33,\sigma=0.12,\kappa=0.16$, and other option pricing model parameters: $r=0, K=100, T=1$. Numerical discretization parameters were chosen as follows: $\Delta x=0.01,\Delta t=0.005$. Since the Variance Gamma process has infinite activity, we employ numerical discretization scheme described in Section 4.2. In what follows, we present various option prices computed by means of the finite-difference numerical scheme described in Section 4 for the linear Black--Scholes $(\rho=0)$ and the Frey--Stremme model $(\rho>0)$ and their jump-diffusion PIDE generalizations.

\begin{figure}
\centering
\includegraphics[width=7truecm]{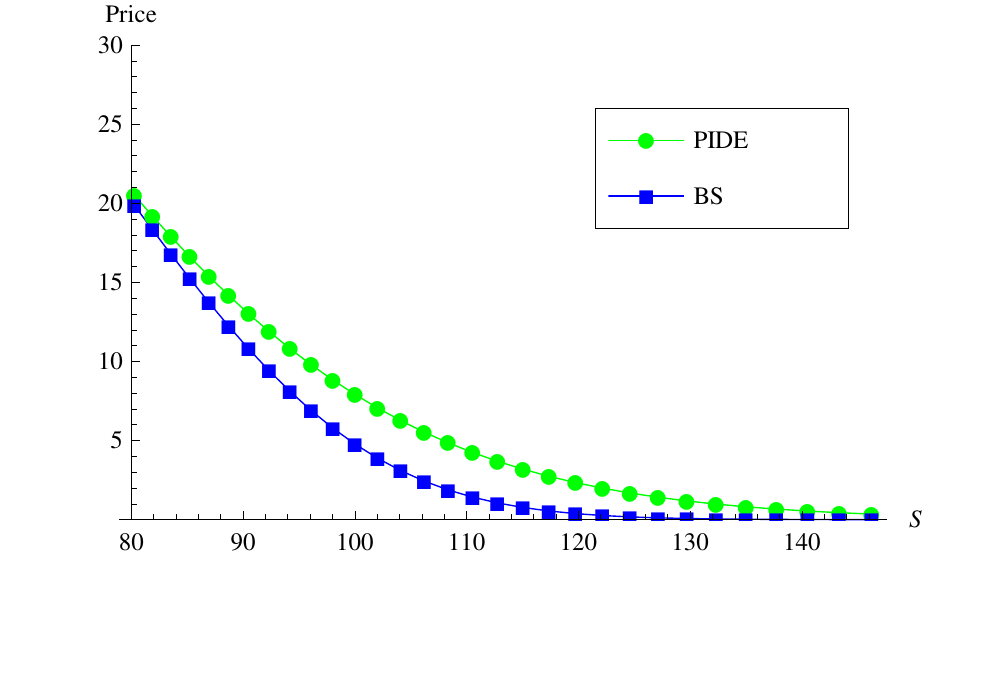}
\ 
\includegraphics[width=7truecm]{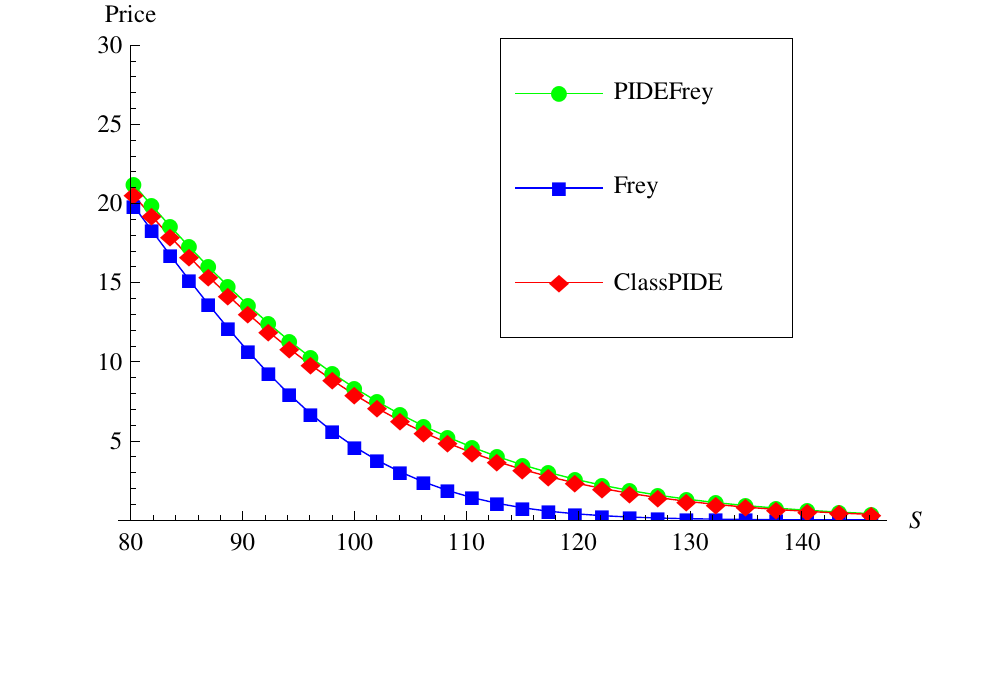}
\caption{Comparison of European put option prices between the classical PIDE and the linear  Black--Scholes model (left). Comparison between the classical PIDE and the Frey--Stremme PIDE  model (right).}
\label{fig-1}
\end{figure}

\begin{figure}
\centering
\includegraphics[width=8truecm]{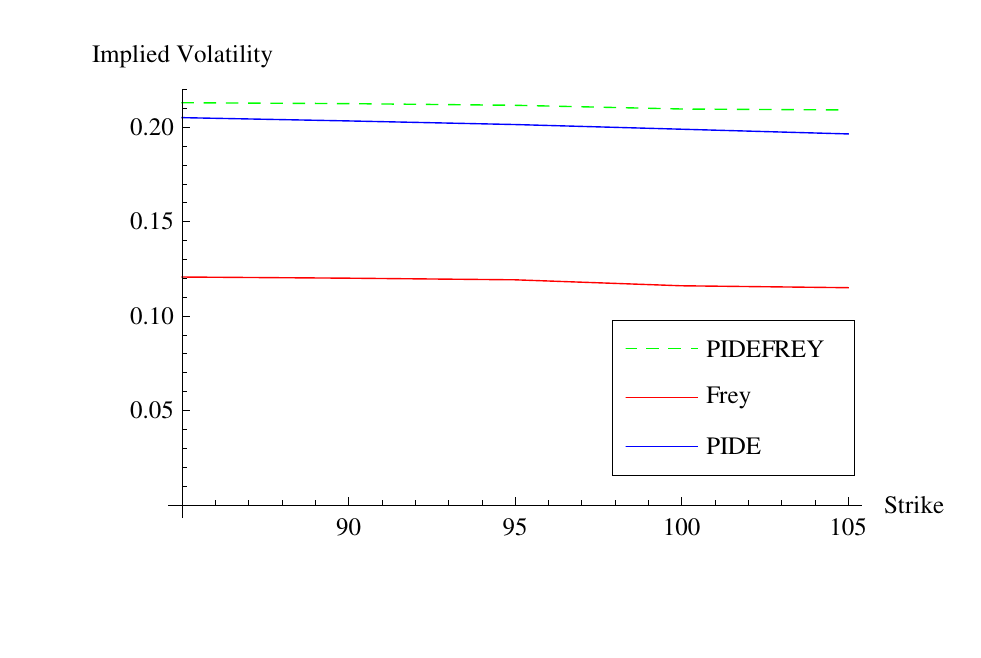}
\caption{Comparison of implied volatilities between the Frey--Stremme model, classical PIDE and Frey--Stremme PIDE  generalization.}
\label{fig-2}
\end{figure}

In Fig. \ref{fig-1} we show comparison of European put option prices between the classical PIDE and the linear  Black--Scholes model, and  comparison between the classical PIDE and the Frey--Stremme PIDE model for the case when the large trader's influence is small,  $\rho=0.001$. In Fig. \ref{fig-2} we plot dependence of the implied volatilities as decreasing functions of the strike price $K$ for the Frey--Stremme model and its PIDE generalizations. We can observe that the implied volatilities for the Frey--Stremme PIDE model is always higher when varying the strike price of the European Put option.

\begin{table}

\centering
\caption{European put option prices $V(0,S)$ for the Black-Scholes and Frey--Stremme models with $\rho=0.001$ and their PIDE generalizations.}
\label{tab1}

\small
\smallskip

\begin{tabular}{l|llll}
     &     B-S         &  F-S                 &  B-S PIDE           &  F-S PIDE  \\
$S$  & $\nu=0, \rho=0$ & $\nu=0, \rho\not=0 $ & $\nu\not=0, \rho=0$ & $\nu\not=0, \rho\not=0$
\\
\hline\hline
61.8783 & 38.1217 & 38.1258 & 38.2297 & 38.8234 \\
67.032 & 32.9691 & 32.9763 & 33.4319 & 34.1889 \\
72.6149 & 27.3972 & 27.4207 & 28.4887 & 29.4425 \\
78.6628 & 21.4275 & 21.5118 & 23.5224 & 24.6911 \\
85.2144 & 15.2547 & 15.4835 & 18.6979 & 20.0701 \\
92.3116 & 9.42895 & 9.85754 & 14.2078 & 15.7321 \\
100. & 4.78444   &  5.32697 & 10.243 & 11.8282 \\
108.329 & 1.88555 & 2.34727 & 6.95353 & 8.48304 \\
117.351 & 0.550422 & 0.814477 & 4.41257 & 5.77178 \\
127.125 & 0.114716 & 0.216426 & 2.60009 & 3.70615 \\
137.713 & 0.016615 & 0.043112 & 1.41444 & 2.2351 \\

\hline
\end{tabular}
\end{table}

\begin{table}
\centering
\caption{European put option prices $V(0,S)$ for the Frey--Stremme and Frey-Stremme PIDE  models for various values of $\rho$.}
\label{tab2}

\small
\smallskip

\begin{tabular}{l|ll|ll|ll}
       &   F-S       & F-S PIDE    &    F-S     &  F-S PIDE   &   F-S       &  F-S PIDE   \\
$S$    &   $\rho=0.1$&  $\rho=0.1$ & $\rho=0.2$ &  $\rho=0.2$ & $\rho=0.3$  &$\rho=0.3$ \\
\hline\hline 
61.8783 & 38.1257 & 38.4958 & 38.1258 & 38.8234 & 38.1373 & 39.2259 \\
67.032 & 32.9759 & 33.7763 & 32.9763 & 34.1889 & 33.019 & 34.6865 \\
72.6149 & 27.4191 & 28.9293 & 27.4207 & 29.4425 & 27.5623 & 30.049 \\
78.6628 & 21.5061 & 24.0698 & 21.5118 & 24.6911 & 21.8893 & 25.4118 \\
85.2144 & 15.4688 & 19.3477 & 15.4835 & 20.0701 & 16.2645 & 20.896 \\
92.3116 & 9.83127 & 14.9344 & 9.85754 & 15.7321 & 11.0916 & 16.6367 \\
100. &   5.29421 & 10.9999 & 5.32697 & 11.8282 & 6.8043 & 12.7672 \\
108.329 & 2.31882 & 7.68096 & 2.34727 & 8.48304 & 3.68338 & 9.4005 \\
117.351 & 0.797286 & 5.05246 & 0.814477 & 5.77178 & 1.72932 & 6.61053 \\
127.125 & 0.209195 & 3.11214 & 0.216426 & 3.70615 & 0.693804 & 4.41995 \\
137.713 & 0.040995 & 1.78547 & 0.043112 & 2.2351 & 0.234949 & 2.79821 \\
\hline
\end{tabular}
\end{table}

Numerical values of option prices for various models and parameter settings  are summarized in Tables~\ref{tab1} and \ref{tab2}. Numerical results confirm our expectation that assuming risk arising from sudden jumps in the underlying asset process yields a higher option price when comparing to the Frey--Stremme model option price.

\begin{figure}
\centering
\includegraphics[width=7truecm]{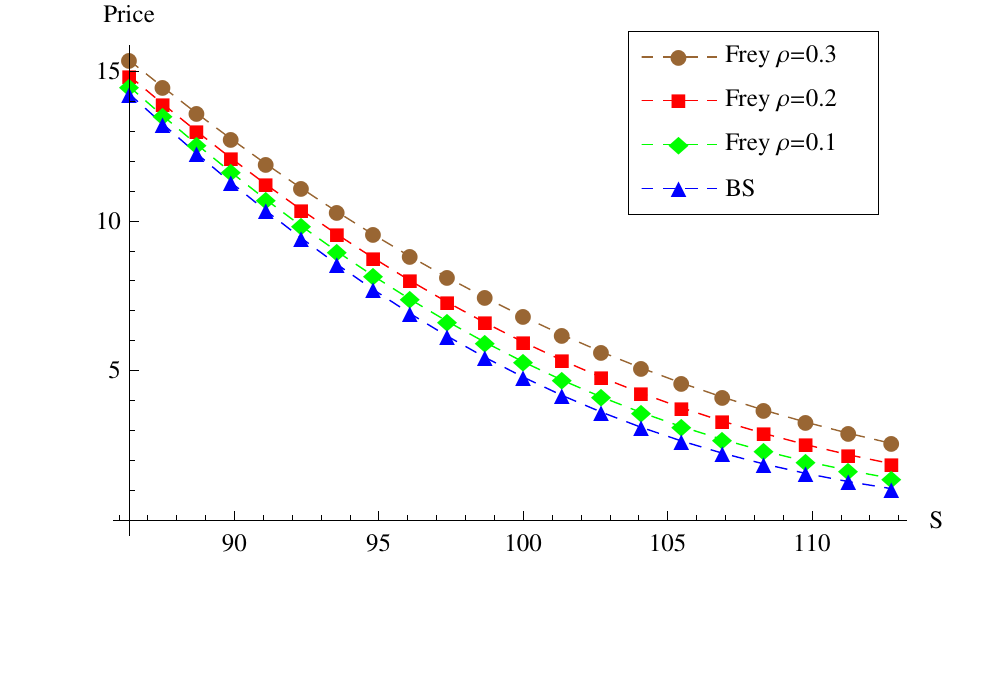}
\ 
\includegraphics[width=7truecm]{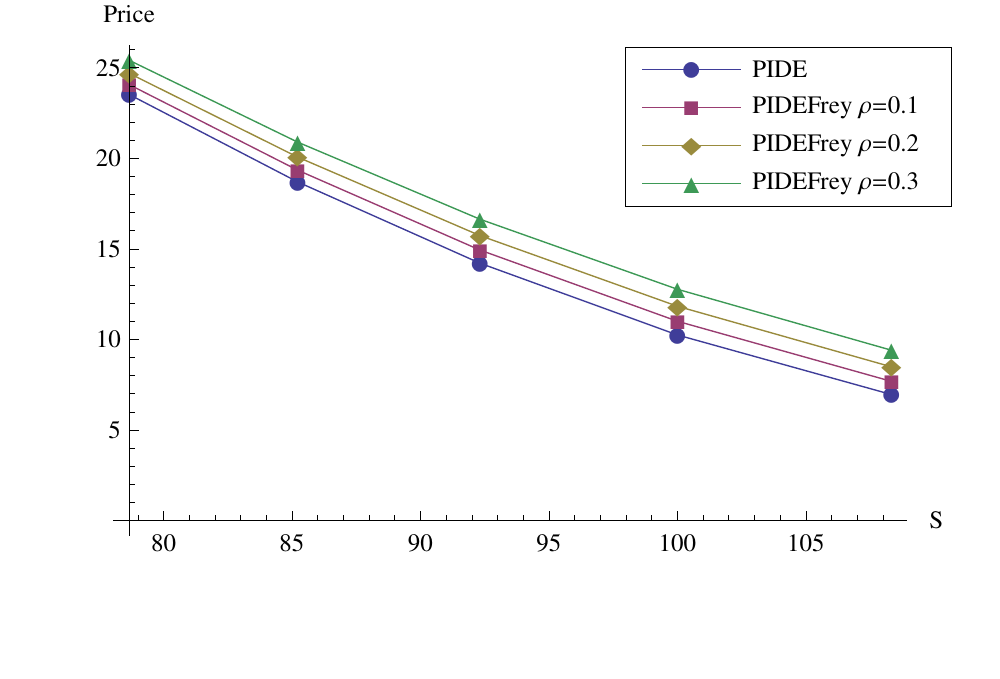}

\caption{Comparison of European put option prices for the Black--Scholes and the Frey--Stremme models (left) and Frey-Stremme PIDE model for various $\rho$.}
\label{fig-3}
\end{figure}

In Fig. \ref{fig-3} (left) we compare European put option prices $V(0,S)$ computed by means of the Black--Scholes and Frey-Stremme models depending on the parameter $\rho$ measuring  influence of a large trader. We can observe that the price of the European put option increases with respect to $\rho$, as expected. Furthermore, the price computed from the Frey--Stremme PIDE model is larger than the one obtained from the linear Black--Scholes equation. Moreover, the price computed from Frey--Stremme PIDE model is higher than the one computed by means of the nonlinear Frey--Stremme model. This is due to the fact that the jump part of the underlying asset process enhances risk, and, consequently increases the option price. Fig. \ref{fig-3} (right) shows comparison of the option prices for the Black--Scholes and Frey--Stremme PIDE model for various values of $\rho$.


\section{Conclusion}

In this paper we investigated a novel nonlinear option pricing model generalizing the Frey--Stremme model under the assumption that the underlying asset price follows a L\'evy stochastic process. We derived the fully-nonlinear PIDE for pricing options under influence of a large trader. We also proposed the hedging strategy minimizing the variance of the tracking error. We derived a semi-implicit finite difference numerical approximation scheme for solving the nonlinear PIDE. We presented various numerical experiments illustrating the influence of the large trader under the L\'evy  process with jumps.

\section{Acknowledgements}

The research was supported by the project CEMAPRE MULTI/00491 financed by FCT/MEC through national funds and the Slovak Research Project VEGA 1/0062/18.

\bibliographystyle{amsplain}

\bibliography{paper}

\end{document}